\documentclass[10pt, conference]{IEEEtran}
\IEEEoverridecommandlockouts
\usepackage{cite}
\usepackage{amsmath,amssymb,amsfonts, amsthm,enumitem}
\usepackage[colorlinks=true,
linkcolor=green!25!black,
urlcolor=blue,
citecolor=blue]{hyperref} 
\usepackage{algorithm, algorithmic}
\usepackage{graphicx}
\usepackage{textcomp}
\usepackage{xcolor}
\usepackage{tikz}
\def\BibTeX{{\rm B\kern-.05em{\sc i\kern-.025em b}\kern-.08em
    T\kern-.1667em\lower.7ex\hbox{E}\kern-.125emX}}

\usepackage{pgfplots}
\usetikzlibrary{arrows,shapes,calc,tikzmark,backgrounds,matrix,decorations.markings}
\usepgfplotslibrary{fillbetween}

\pgfplotsset{compat=1.3}

\usepackage{relsize}
\tikzset{fontscale/.style = {font=\relsize{#1}}
    }

\definecolor{lavander}{cmyk}{0,0.48,0,0}
\definecolor{violet}{cmyk}{0.79,0.88,0,0}
\definecolor{burntorange}{cmyk}{0,0.52,1,0}

\definecolor{asuorange}{rgb}{1,0.699,0.0625}
\definecolor{asured}{rgb}{0.598,0,0.199}
\definecolor{asuborder}{rgb}{0.953,0.484,0}
\definecolor{asugrey}{rgb}{0.309,0.332,0.340}
\definecolor{asublue}{rgb}{0,0.555,0.836}
\definecolor{asugold}{rgb}{1,0.777,0.008}

\DeclareUnicodeCharacter{2212}{−}
\usepgfplotslibrary{groupplots,dateplot}
\usetikzlibrary{patterns,shapes.arrows}
\pgfplotsset{compat=newest}

\usepackage{tikz,pgfplots}
\usepgfplotslibrary{groupplots}

\newcommand{\matr}[1]{\mathbf{#1}}
\usepackage{yhmath}


\usepackage{shortcuts}

\newtheorem{Lemma}{Lemma}
\newtheorem{prob}{Problem}

\newtheorem{Theorem}{Theorem}
\newtheorem{Def}{Definition}

\newtheorem{assumption}{H\!\!}

    \makeatletter
    \def\multilimits@{\bgroup
  \Let@
  \restore@math@cr
  \default@tag
 \baselineskip\fontdimen10 \scriptfont\tw@
 \advance\baselineskip\fontdimen12 \scriptfont\tw@
 \lineskip\thr@@\fontdimen8 \scriptfont\thr@@
 \lineskiplimit\lineskip
 \vbox\bgroup\ialign\bgroup\hfil$\m@th\scriptstyle{##}$\hfil\crcr}
    \def\Sb{_\multilimits@}
    \def\endSb{\crcr\egroup\egroup\egroup}
\makeatother

\makeatletter
\DeclareRobustCommand*\cal{\@fontswitch\relax\mathcal}
\makeatother

\begin{document}

\title{On Detecting Low-pass Graph Signals\\under Partial Observations \thanks{This work is supported in part by  HKRGC Project \#24203520.}
}

\author{\IEEEauthorblockN{Hoang-Son Nguyen}
\IEEEauthorblockA{\textit{Dept. of Computer Science \& Engineering} \\
\textit{The Chinese University of Hong Kong}\\
sonngh.00@link.cuhk.edu.hk}
\and
\IEEEauthorblockN{Hoi-To Wai}
\IEEEauthorblockA{\textit{Dept. of Systems Engineering \& Engineering Management} \\
\textit{The Chinese University of Hong Kong}\\
htwai@se.cuhk.edu.hk}
}

\maketitle

\begin{abstract}
        The application of graph signal processing (GSP) on partially observed graph signals with missing nodes has gained attention recently. This is because processing data from large graphs are difficult, if not impossible due to the lack of availability of full observations. Many prior works have been developed using the assumption that the generated graph signals are smooth or low pass filtered. This paper treats a blind graph filter detection problem under this context. We propose a detector that certifies whether the partially observed graph signals are low pass filtered, without requiring the graph topology knowledge. As an example application, our detector leads to a pre-screening method to filter out non low pass signals and thus robustify the prior GSP algorithms. We also bound the sample complexity of our detector in terms of the class of filters, number of observed nodes, etc. Numerical experiments verify the efficacy of our method.
\end{abstract}


\section{Introduction}
An important goal of graph signal processing (GSP) \cite{ortega2018graph} is to extract insights from complex network data. Using graph shift operator (GSO), graph filters \& signals as the underlying constructs, GSP has led to many theoretically justified graph learning methods \cite{mateos2019connecting, dong2019learning}, e.g., prior works showed how to estimate the structure of weather \cite{thanou2017learning} and brain networks \cite{huang2018graph}. Meanwhile, the overwhelming size of complex networks has necessitated practical methods to consider the \emph{partial observation} setting where a fraction of nodes are never observed. 

The partial observation setting may break a number of properties such as structure of eigenvectors, smoothness of graph signals, etc., that are necessary for graph learning. To this end, the early work \cite{chandrasekaran2012latent} proposed to exploit the `low-rank+sparse' structure in the precision matrix of partially observed graph signal. Subsequent work such as \cite{buciulea2022learning} proposed a graph learning criterion using smoothness of graph signals, \cite{jalali2011learning} considered time-series data, \cite{matta2020graph} considered a linear influence model, and \cite{hendrickx2018identifiability, santos2023learning} focused on identifiability of network dynamical systems. Additionally, the authors have studied graph feature learning from partial observations, such as community \cite{wai2022partial}, central nodes \cite{he2023central}. In the above works, a common assumption made is that the graph signals are \emph{smooth}, or more generally, generated from a network process that can be modeled as exciting a \emph{low-pass graph filter} \cite{ramakrishna2020userguide}. 

While the \emph{low-pass graph filter} assumption can be motivated by modeling network processes from social-physical aspects (e.g., \cite{ramakrishna2020userguide}, \cite{degroot1974reaching}), the latter often requires prior knowledge on the given dataset. In the absence of prior knowledge or when the dataset is corrupted, applying GSP methods may lead to unexpected results. 
Under this context, it is natural to ask:
    \emph{Do we know if a dataset of partially observed graph signals is generated from a low-pass filter, without knowing the underlying graph beforehand?}
Addressing the question gives a certificate \emph{prior to} applying the mentioned methods on partially observed signals and guarantees reliable outcomes. 

Our plan is to build on the authors' prior work \cite{zhang2023detecting}, which tackled a similar detection problem but was focused on fully observed graph signals. Particularly, as the detection problem is ill-posed in general since smoothness/low-pass-ness are defined with respect to the graph itself, \cite{zhang2023detecting} focuses on a simplified case where the graph is known to be modular \cite{girvan2002community}, a common feature for graphs found in networked systems. It then derives a detector based on the clusterizability of principal components, i.e., spectral pattern, for observed graph signals. 

For partially observed graph signals, the challenge lies on how to account for the effect of missing nodes on the \emph{observed} spectral pattern. To this end, our contributions are: 
\begin{itemize}[leftmargin=*]
    \item We show that the $K$-means score detector in \cite{zhang2023detecting} can correctly detect the spectral pattern of \emph{partially observed} low-pass graph signals. Though the latter also exhibit spectral pattern that distinguishes itself from any non-low-pass signals, we prove that the sampling complexity critically depends on the number of observed nodes. 
    \item We demonstrate that the proposed detector can be used as a pre-screening procedure to robustify community detection from partially observed graph signals.
\end{itemize}
The rest of this paper is structured as follows. Section~\ref{sec:problem} describes the partial observation setting and formulates the blind detection problem. Section~\ref{sec:method} develops the proposed method and reports its sample complexity. Finally, Section~\ref{sec:exp} presents results from preliminary numerical experiments. 

\textbf{Notations}. We use $||\cdot ||_2$ to denote spectral norm for matrices and Euclidean norm for vectors, and $||\cdot||_{\rm F}$ to denote Frobenius matrix norm. 
For a symmetric matrix $\matr{X}$, $\lambda_{i}(\matr{X})$ denotes the $i^{th}$ smallest eigenvalue of a matrix.

\section{Problem Statement} \label{sec:problem}

Consider an undirected, connected $N$-node graph $G = (V, E)$ where $V = \{1, ..., N\}$ and $E \subseteq V \times V$. The graph can be represented as an adjacency matrix $\matr{A} \in \{0, 1\}^{N \times N}$, a Laplacian matrix $\matr{L} = \matr{D} - \matr{A}$  where $\matr{D} = \diag(\matr{A}\matr{1})$, or a normalized Laplacian matrix $\matr{L}_{\text{norm}} = \matr{I} - \matr{D}^{-1/2}\matr{A}\matr{D}^{-1/2}$. The aforementioned matrix representations of $G$ qualify as graph shift operators (GSO), which are any symmetric matrix $\matr{S} \in \mathbb{R}^{N \times N}$ such that $\matr{S}_{ij} \neq 0$ only if $(i, j) \in E$. A GSO admits an eigendecomposition $\matr{S} = \matr{V}\matr{\Lambda}\matr{V}^{\top}$, where the columns of $\matr{V} = [ \matr{v}_1, \cdots, \matr{v}_N ]$ are the orthonormal eigenvectors associated with eigenvalues sorted in ascending order, and $\matr{\Lambda}$ is a diagonal matrix of eigenvalues, also known as the graph frequencies. We focus on the case of $\matr{S} = \matr{L}_{\rm norm}$. 

The graph filter is defined as a polynomial of the GSO:
\begin{align} \label{eq:filter} \textstyle
    \mathcal{H}(\matr{S}) = \sum_{t = 0}^{T}h_t\matr{S}^{t} = \matr{V}h(\matr{\Lambda})\matr{V}^{\top},
\end{align}
where $\{ h_t \}_{t = 0}^{T-1}$ are the filter coefficients. The latter also defines the frequency response function: $h(\lambda) = \sum_{t = 0}^{T}h_t\lambda^{t}$ and $h(\matr{\Lambda}) = \diag(h(\lambda_1), ..., h(\lambda_n))$. {For simplicity, we assume the magnitudes of frequency responses to be distinct, i.e. $|h(\lambda_i)| \neq |h(\lambda_j)|$ for all $i \neq j$.} By sorting the magnitude of frequency responses in descending order as $|h_1| > ... > |h_N|$, the graph filter operator can be written as $\mathcal{H}(\matr{S}) = \matr{U}\matr{h}\matr{U}^{\top}$, where $\matr{h} = \diag(h_1, ..., h_N)$ and $\matr{U}$ is accordingly the column re-ordered version of $\matr{V}$. 

A graph signal on $G$ is represented as a $N$-dimensional vector that is the output of a graph filter \eqref{eq:filter}:
\begin{align} \label{eq:network_process} \textstyle
    \matr{y} = \mathcal{H}(\matr{S})\matr{x} + \matr{w}.
\end{align}
The $i$th element of $\matr{y}$ denotes the signal on node $i \in V$,
subjected to the excitation graph signal $\matr{x} \in \mathbb{R}^{N}$, and $\matr{w} \in \mathbb{R}^{n}$ is a zero-mean white noise with $\mathbb{E}[\matr{w}\matr{w}^{\top}] = \sigma^2 \matr{I}$ for $\sigma^2 \geq 0$. The observed signal $\matr{y}$ is assumed to be stationary, i.e., $\mathbb{E}[\matr{x}] = \matr{0}$ and $\mathbb{E}[\matr{x}\matr{x}^{\top}] = \matr{I}$ \cite{perraudin2017stationary}, \cite{marques2017stationary} for simplicity; however, our analysis can be extended to the non-stationary case of $\mathbb{E}[\matr{x}\matr{x}^{\top}] \neq \matr{I}$.

The graph signal $\matr{y}$ in \eqref{eq:network_process} can also be modeled as the output of a network process. Prior works in GSP have suggested to categorize network process according to their frequency response. Among others, an important class of graph filters is the low-pass graph filters \cite{sandryhaila2013discrete,  ramakrishna2020userguide}, which is defined by:
\begin{Def} \label{def:lpf}
A graph filter $\mathcal{H}(\cdot)$ is said to be $K$-low-pass if
\begin{align}
    \eta_{K} := \frac{\max_{i = K+1, ..., N}|h(\lambda_i)|}{\min_{i = 1, ..., K}|h(\lambda_i)|} < 1, \label{eq:lowpassK}
\end{align}
where $K$ is cut-off frequency, and $\eta_{K}$ is the sharpness of $\mathcal{H}$.
\end{Def}
\noindent From the definition, a low-pass graph filter retains (resp.~attenuates) the energy of the excitation graph signal at low (resp.~high) frequencies. A graph signal is said to be low-pass if it is the output of a low-pass graph filter.


We further consider the scenario when the graph signals in \eqref{eq:network_process} can only be partially observed. Without loss of generality, we assume that the first $n$ nodes are observed and denote 
\begin{equation} \label{eq:yo}
    \matr{y}_o = [\matr{I}_{n \times n}, \matr{0}_{n \times (N-n)}] \, \matr{y} =: \matr{E}_{o}\matr{y} .
\end{equation}
As mentioned in the Introduction, the application of GSP on partial observations has gained popularity as the model arises naturally for large graphs where it is difficult to obtain observations on every nodes. 
Under this context, GSP applications such as graph learning \cite{buciulea2022learning}, community detection \cite{wai2022partial} have exploited the \emph{smoothness} property and motivate the latter by modeling the graph signal observations as \emph{low pass signals}.


We depart from the above works and inquire if the \emph{smoothness} property is valid for a given dataset. This leads to the blind low-pass graph filter detection problem:
\begin{prob} \label{prob:detect}
Given the parameter $K$ and a set of partially observed graph signals [cf.~\eqref{eq:network_process}, \eqref{eq:yo}], determine if the underlying graph filter is $K$-low-pass or not [cf.~Definition~\ref{def:lpf}]. We denote the null hypothesis ${\cal T}_0$ (resp.~alternative hypothesis ${\cal T}_1$) as `${\cal H}({\bf S})$ is (resp.~not) $K$-low-pass'.
\end{prob}
\noindent Notice that Problem~\ref{prob:detect} serves as a data-driven certificate to the successful applications of the prior GSP works. 

There are two challenges in solving Problem~\ref{prob:detect}: (i) the graph topology or the GSO ${\bf S}$ is unknown, (ii) the graph signals are partially observed where ${\bf E}_o$ is unknown. Either challenge has made it impossible to verify Definition~\ref{def:lpf} directly.
Our prior work \cite{zhang2023detecting} proposed to narrow down the detection problem w.r.t.~arbitrary graphs to the class of $K$-modular graphs \cite{girvan2002community} with $K$ densely connected components. The number of densely connected components naturally determines the parameter $K$ for the low-pass filter. It then exploits the spectral pattern of graphs to formulate a $K$-means score detector\footnote{We remark that when $K=1$, i.e., the graph contains only one dense component, applying the Perron Frobenius theorem \cite{he2021lowpass} suffices to detect the $1$-low-pass graph signals. Here, we shall focus on the case of $K \geq 2$.}. The detector is proven to produce accurate result under mild assumptions on the noise statistics and graph filter properties.

This paper aims to extend the aforementioned detector in \cite{zhang2023detecting} to the partial observation context. Interestingly, we show that the $K$-means score detector is still robust in this scenario, whose performance loss depends naturally with the ratio $n/N$.

\section{Low-pass Detection with Partial Observations}\label{sec:method}
This section develops a detector for Problem~\ref{prob:detect} under the partial observation setting. Our development begins by investigating the \emph{covariance matrix} of partially observed signals:
\begin{equation}
\begin{split}
    \matr{C}_{o} = \mathbb{E}[\matr{y}_{o, m}\matr{y}_{o, m}^{\top}] & = \matr{V}_{o}h(\matr{\Lambda})^2\matr{V}_{o}^{\top} + \sigma^2\matr{I}\\
    & = \matr{U}_{o} {\bf h}^2 \matr{U}_{o}^{\top} + \sigma^2\matr{I},
\end{split}
\end{equation}
where we have used $\matr{y}_{o, m}$ to denote the $m$th realization of the partially observed signal in \eqref{eq:yo}. We have defined the \emph{row-sampled} eigenmatrices $\matr{V}_o = \matr{E}_o \matr{V}, \matr{U}_o = \matr{E}_o \matr{U} \in \mathbb{R}^{n \times N}$. The noiseless covariance is $\overline{\matr{C}}_{o} = \matr{C}_{o} - \sigma^2 \matr{I} = \matr{U}_{o} {\bf h}^2 \matr{U}_{o}^{\top}$.

Following the insight from \cite{wai2022partial}, we note that when the graph filter has a \emph{sharp} cut-off (e.g., $\eta_K \ll 1$ under ${\cal T}_0$), the following approximation holds
\begin{equation} \label{eq:approx_Co}
    \overline{\matr{C}}_{o} \approx \matr{U}_{o,K} {\bf h}_K^2 \matr{U}_{o,K}^{\top},
\end{equation}
where $\matr{U}_{o,K}$ takes the $K$ left-most column vectors from $\matr{U}_o$. Under ${\cal T}_0$, the matrix $\matr{U}_{o,K}$ corresponds to the row-sampled and column permuted version of $\matr{V}_{K} = [\matr{v}_1, ..., \matr{v}_{K}]$. To this end, the \emph{row vectors} of $\matr{V}_K$ are clusterizable when $G$ is $K$-modular.
Meanwhile, under ${\cal T}_1$ when the graph filter is \emph{not} $K$-low-pass, ${\bf U}_{o,K}$ corresponds to the row sampled versions of the bulk eigenvectors $\{\matr{v}_{K+1}, ..., \matr{v}_{N}\}$ which are not clusterizable \cite{vonLuxburg2007sc}. Together, they motivate the following $K$-means score: for any $\matr{N} \in \mathbb{R}^{N \times K}$, we denote
\begin{equation}
\begin{split}
    \mathbb{K}^{*}(\matr{N}) & \textstyle := \min_{\mathcal{C}} \, \mathbb{K}(\matr{N}, \mathcal{C}), \\
    \mathbb{K}(\matr{N}, \mathcal{C}) & \textstyle := \sum_{k=1}^{K}\sum_{i \in \mathcal{C}_k}||\matr{n}^{\rm row}_i - \frac{1}{|\mathcal{C}_{k}|}\sum_{j \in \mathcal{C}_{k}}\matr{n}^{\rm row}_{j}||^2_2,
\end{split}
\end{equation}
where ${\cal C} = \{ {\cal C}_1, \ldots, {\cal C}_K \}$ is a set of non-overlapping partition for $\{1,...,N\}$ and $\matr{n}^{\rm row}_i$ denotes the $i$th row vector of $\matr{N}$.

Define the sampled covariance matrix $\widehat{ \matr{C} }_o : = (1/M)\sum_{m=1}^{M}\matr{y}_{o, m}\matr{y}_{o, m}^{\top}$ and its top-$K$ eigenvectors are stacked up as $\widehat{\matr{Q}}_K$.
Following the insights from \cite{wai2022partial, zhang2023detecting} and observe that $\widehat{\matr{Q}}_K \approx \matr{U}_{o,K}$ when $n$ is sufficiently close to $N$, we propose to tackle Problem~\ref{prob:detect} by detecting ${\cal T}_0$/${\cal T}_1$ based on $\mathbb{K}^*( \widehat{\bf Q}_K )$. From the above discussions, $\mathbb{K}^*( \widehat{\bf Q}_K )$ will be small (resp.~large) when the graph filter is (resp.~not) $K$-low-pass. This motivates the proposed detector in Algorithm~\ref{alg:detector}.


\begin{algorithm}[t]
\caption{Low-pass Detection with Partial Observations} \label{alg:detector}
\begin{algorithmic}[1]
\STATE \textbf{Input:} Partially observed graph signals $\{\matr{y}_{o, m}\}_{m=1}^{M}$, no. of clusters $K \geq 2$, detection threshold $\delta > 0$.
\STATE Calculate $\widehat{\matr{C}}_{o} := (1/M)\sum_{m=1}^{M}\matr{y}_{o, m}\matr{y}_{o, m}^{\top}$.
\STATE Compute the top-$K$ eigenvectors $\widehat{\matr{Q}}_{K} \in \mathbb{R}^{n \times K}$ of $\widehat{\matr{C}}_{o}$.
\STATE \textbf{Output:} $\widehat{\mathcal{T}} = {\cal T}_0$ if $\mathbb{K}^{*}(\widehat{\matr{Q}}_K) < \delta$; or $\widehat{\mathcal{T}} = {\cal T}_1$ otherwise.
\end{algorithmic} 
\end{algorithm}


\subsection{Performance Analysis and Theoretical Insights}
We next present the analysis on the finite-sample performance of Algorithm~\ref{alg:detector}. In addition to verifying the correctness of the detector, our analysis shall demonstrate the favorable conditions where the detector is effective. Note that there are multiple sources of error that need to be controlled carefully. For instance, the approximation in \eqref{eq:approx_Co} is not exact, the columns of $\matr{U}_{o,K}$ are not orthogonal, etc. 

To set up the analysis, we require the following condition on spectral gap of the covariance matrix:
\begin{assumption} \label{assu:spectral_gap}
With probability at least $1-\delta_{\rm gap}$, there exists $\rho_{\rm gap}$ such that $\lambda_{n-K-1}(\overline{\matr{C}}_{o}) - \lambda_{n-K}(\overline{\matr{C}}_{o}) - ||\widehat{\matr{C}}_{o} -\overline{\matr{C}}_{o}||_{2} \geq \rho_{\rm gap} > 0$.
\end{assumption}
\noindent The above can be satisfied when $\widehat{\matr{C}}_{o}$ is sufficiently close to $\overline{\matr{C}}_{o}$, e.g., when sufficient number of samples are observed and the noise level $\sigma^2$ is small, and the noiseless covariance $\overline{\matr{C}}_o$ is approximately rank $K$. We also let: 
\begin{assumption} \label{assu:sharp}
The graph filter ${\cal H}( \matr{S} )$ is at least $\eta$-sharp and $\gamma$-flat:
\begin{equation}
\frac{ \max_{i=K+1,...,N} |h_i| }{ \min_{i=1,...,K} |h_i| } \leq \eta < 1, ~~ \frac{\max_{1 \leq i \leq K}h_i^2}{\min_{1 \leq j \leq K}h_j^2} \leq \gamma.
\end{equation}
\end{assumption}
\noindent The above specifies the class of graph filters that we detect. Notice if the graph filter is $K$-low-pass, then the above $\eta$ takes the same role as $\eta_K$ in \eqref{eq:lowpassK}. 

As mentioned in the previous section, the proposed detector relies on the clusterizability of the top-$K$ eigenvectors for the normalized Laplacian in $K$-modular graphs. To obtain theoretical insights, we assume that the full graph $G$ is generated from the stochastic block model (\text{SBM}) with: 
\begin{assumption} \label{assu:sbm}
We have $G \sim \text{SBM}( N, K , r, p )$ with $p \geq r > 0$, $p/K + r \geq (32 \log N + 1)/N$. 
\end{assumption}
\noindent By $G \sim \text{SBM}( N, K , r, p )$, we denote a random graph with $N$ nodes equally partitioned into $K$ blocks, described by a membership matrix $\mathbf{Z} \in \{0, 1\}^{N \times K}$ such that $\mathbf{Z}_{ij} = 1$ if and only if node $i$ is in block $j$, and a connectivity matrix $\mathbf{B} \in [0, 1]^{K \times K}$, whose entries $\matr{B}_{ij}$ being the probability of edges between nodes in block $i$ and block $j$. The parameters $r,p$ describes the connectivity such that $\matr{B} = p\matr{I} + r\matr{1}\matr{1}^{\top}$.
We also assume the following on the \emph{bulk eigenvectors} of $\matr{L}_{\rm norm}$:
\begin{assumption} \label{assu:csbm}
With probability at least $1-\delta_{\rm \text{SBM}}$, there exists $c_{\text{\text{SBM}}} > 0$ independent of $N,r,p$ with $\min_{ l=K+1,...,N} \mathbb{K}^{*}(\matr{v}_{l}) \geq c_{\text{\text{SBM}}}$.
\end{assumption}
\noindent Note that H\ref{assu:csbm} is observed for $G \sim \text{SBM}(N,K,r,p)$ empirically \cite{zhang2023detecting}, yet it remains an open conjecture to be verified theoretically.
With H\ref{assu:sbm}, H\ref{assu:csbm}, it is easy to deduce that $\mathbb{K}^{*}(\matr{V}_{K}) = {\cal O}( \log N / N )$ \cite{rohe2011sbm}, while the $K$-means score for the bulk eigenvectors is bounded away by $c_{\text{\text{SBM}}} > 0$.

Let ${\cal T}_{\rm gnd} \in \{ {\cal T}_0, {\cal T}_1 \}$ be the ground truth hypothesis. 
Our main analytical result is summarized below:
\begin{Theorem} \label{th:main}
Under H\ref{assu:spectral_gap}, H\ref{assu:sharp}, H\ref{assu:sbm}, H\ref{assu:csbm}.
Suppose that the following threshold-dependent term satisfies
\begin{align*}
    \Tilde{\delta}_{\rm min} := \min\Bigg\{&\delta - \sqrt{\frac{N}{n}}\sqrt{\frac{1225K^3 \log N}{p(N-K)}},\\&\sqrt{\frac{N}{n}}\sqrt{c_{\rm \text{SBM}} - \frac{2450 K^3 \log N}{p(N-K)}} - \delta\Bigg\} > 0,
\end{align*}
and that 
\[ 
\tilde{\sigma} := \frac{ \rho_{\rm gap} (\Tilde{\delta}_{\rm min} - \sqrt{K} ( ||\matr{I} - \matr{R}_{K}||_{\rm 2} + 6 \gamma \eta ) )  }{ 2\sqrt{K} } - \sigma^2 > 0,
\]
where $\matr{R}_K \in \mathbb{R}^{K \times K}$ is an upper triangular matrix in the QR factorization of $\matr{U}_{o,K} = \sqrt{N/n} \matr{Q}_K \matr{R}_K$.
If the number of samples $M$ satisfies
\begin{align} \label{eq:sample_no}
    &\sqrt{\frac{M}{\log M}} \geq \frac{\sqrt{2}c_1 \tr(\overline{\matr{C}}_{o})}{ \tilde{\sigma} }
\end{align}
where $c_1$ is a constant independent of $N,M$, then we have
\begin{align}
    \mathbb{P}(\widehat{\mathcal{T}} = \mathcal{T}_{\rm gnd}) \geq 1 - 4/N - 5/M - \delta_{\rm gap} - \delta_{\rm \text{SBM}}.
\end{align}
\end{Theorem}
\noindent The above result considers randomnesses in the graph signals generation \eqref{eq:network_process} and the \text{SBM} graph properties in H\ref{assu:sbm}, H\ref{assu:csbm}.

The theorem asserts that when $\Tilde{\delta}_{\rm min} > 0$, $\tilde{\sigma} > 0$, then with a sufficiently large number of samples, Algorithm~\ref{alg:detector} will return a correct detection with high probability as $N,M \to \infty$. To satisfy $\Tilde{\delta}_{\rm min} > 0$, as $c_{\rm \text{SBM}} = \Theta(1)$, the requirement can be fulfilled with $\delta = \Theta( \sqrt{N/n} )$. Furthermore, to satisfy $\tilde{\sigma} > 0$, we require two criterion: (i) the noise level $\sigma^2$ is sufficiently small, (ii) the filter constant $\gamma \eta$, and the factor $\| \matr{I} - \matr{R}_K \|_2$ are smaller than ${\cal O}( \tilde{\delta}_{\rm min} )$. Note that $\| \matr{I} - \matr{R}_K \|_2$ decreases to $0$ as $n \to N$; see Fig.~\ref{fig:Rk} for illustration.
Finally, we note that the sample complexity, i.e., minimal $M$ needed to satisfy \eqref{eq:sample_no}, is proportional to $\tilde{\sigma}^{-1}$. From the above discussions, $\tilde{\sigma}^{-1}$ is reduced when the graph filters to be detected are \emph{sharp and flat}, i.e., $\eta \ll 1, \gamma \approx 1$, and the number of observed nodes is large enough $n \to N$.
\pgfplotsset{every tick label/.append style={font=\small}}
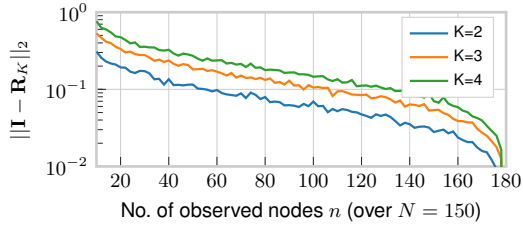
\begin{figure}[t]
\centering
  {\sf \resizebox{.8\linewidth}{!}{
\begin{tikzpicture}

\definecolor{darkgray176}{RGB}{176,176,176}
\definecolor{darkorange25512714}{RGB}{255,127,14}
\definecolor{forestgreen4416044}{RGB}{44,160,44}
\definecolor{steelblue31119180}{RGB}{31,119,180}

\begin{axis}[
axis line style={white!80!black},
legend cell align={left},
legend style={
  fill opacity=0.7,
  draw opacity=0.9,
  text opacity=1.0,
  legend pos = north east,
  draw=white!80!black,
  font = \normalsize
},
log basis y={10},
tick align=inside,
x grid style={white!80!black},
xlabel={\small No. of observed nodes \(\displaystyle n\) (over $N=150$)},
xmajorgrids,
xmajorticks=true,
xmin=10, xmax=180,
xminorgrids,
xtick style={color=white!15!black},
xtick pos=left,
y grid style={white!80!black},
ylabel={\small \(\displaystyle{||\matr{I}-\matr{R}_{K}||_{\rm 2}}\)},
ymajorgrids,
ymajorticks=true,
ymin=0.01, ymax=1,
ymode=log,
ytick style={color=white!15!black},
yticklabel style={font=\small},
ytick pos=left,
width = 8cm,
height = 4cm,
axis line style = very thick
]
\addlegendentry{\scriptsize K=2}
\addplot [line width=1pt, steelblue31119180]
table {%
10 0.308891926695174
12 0.256429318330918
14 0.22909784507943
16 0.216139048449805
18 0.209374296693221
20 0.192151372810452
22 0.190994409690096
24 0.169617892445037
26 0.162781073518169
28 0.173529118204397
30 0.172043431103187
32 0.160789721897486
34 0.134892746783567
36 0.136775436380059
38 0.115227807991333
40 0.134989315635681
42 0.113505961188883
44 0.114390478000555
46 0.112529267582043
48 0.11474854333867
50 0.115362292246973
52 0.102212031696105
54 0.1013132760754
56 0.098747538225892
58 0.0941024633920517
60 0.0978154398547977
62 0.0895081839487836
64 0.0860044546816199
66 0.0843553679123752
68 0.0825420551179798
70 0.0886094753840178
72 0.0770846595211134
74 0.0905522576616397
76 0.0849368494620757
78 0.0746668296111455
80 0.0794361744126487
82 0.0681353500473242
84 0.0692885404261315
86 0.0682432675444049
88 0.0619850392415309
90 0.0623305122828141
92 0.061349314295068
94 0.0589761444324206
96 0.0646715072168334
98 0.0604892542302265
100 0.0694018598308289
102 0.0603520516470166
104 0.0609812859410818
106 0.0504992344587665
108 0.0556401819659879
110 0.0520141841082774
112 0.0553389997314319
114 0.053990677617722
116 0.0510747265869098
118 0.0502316336662456
120 0.0472146542483131
122 0.04518143677533
124 0.0426908464918859
126 0.0473809761649008
128 0.0451780559073666
130 0.0470953612317695
132 0.0437750307442552
134 0.0387730232025651
136 0.0340535953594979
138 0.0368042606457045
140 0.0353355453918993
142 0.0316932760808116
144 0.0327189286670382
146 0.0329742573734709
148 0.0318790909895483
150 0.0279307763546475
152 0.0290389171670087
154 0.0300086573936117
156 0.0269055325197947
158 0.0283454259925556
160 0.0237460460731167
162 0.0219542096749471
164 0.0206979901379143
166 0.0206715397029071
168 0.0176434940783318
170 0.0158599149353179
172 0.015257554010854
174 0.0121491328974469
176 0.00924812254636931
178 0.00734636127463376
180 8.78833948765746e-16
};
\addlegendentry{\scriptsize K=3}
\addplot [line width=1pt, darkorange25512714]
table {%
10 0.530629702502707
12 0.465647810420656
14 0.414668444757654
16 0.396932299867497
18 0.3556413117517
20 0.333232283698591
22 0.299933099853205
24 0.308761585497591
26 0.286158457220426
28 0.278372054104404
30 0.278455370256637
32 0.257361113263585
34 0.234733762875376
36 0.242861218523232
38 0.226291880592787
40 0.236972565117801
42 0.21393013248176
44 0.202158886501352
46 0.212039665260163
48 0.206357469086636
50 0.181197726760679
52 0.182276634181086
54 0.18461169366048
56 0.173138017521982
58 0.176687320887873
60 0.169436268445571
62 0.161220499591735
64 0.151045810482958
66 0.162738574397149
68 0.141859812671611
70 0.150564148272368
72 0.140530554020812
74 0.142509087835117
76 0.135928458897013
78 0.140954835855142
80 0.132371436216715
82 0.138532033361108
84 0.121516346608241
86 0.124325059832205
88 0.115433334039212
90 0.129457378153308
92 0.111833491143949
94 0.113125270071552
96 0.0995187469993611
98 0.113534615862364
100 0.10520794450214
102 0.107031047672158
104 0.103999368802002
106 0.104653462415196
108 0.0818261846029172
110 0.0940903341997531
112 0.0948397094357302
114 0.0903211888070843
116 0.0856264458768559
118 0.0849701696923897
120 0.084175309244669
122 0.0854554547358712
124 0.0765742928869685
126 0.0778312429671969
128 0.0813554594220815
130 0.0791246852292083
132 0.0740581617116876
134 0.0672826804415078
136 0.0715716041807873
138 0.0594939408066062
140 0.0637356160476083
142 0.0627462362450512
144 0.0608841141294805
146 0.0642753677019051
148 0.0529786594567288
150 0.0537690183997307
152 0.0542698300156789
154 0.0490793801983022
156 0.0468099576879048
158 0.0415158716677072
160 0.0388403209210503
162 0.0385632497606956
164 0.0372724229548882
166 0.0370977937761885
168 0.0308749668105716
170 0.0273366555259397
172 0.024972633480094
174 0.021210546125844
176 0.0178752219232412
178 0.0125557755149569
180 9.54210563596271e-16
};
\addlegendentry{\scriptsize K=4}
\addplot [line width=1pt, forestgreen4416044]
table {%
10 0.764523832459012
12 0.627645919130391
14 0.597813564801574
16 0.545691694272978
18 0.485192651506812
20 0.474126103410377
22 0.436440244939905
24 0.394830908314011
26 0.405582293484858
28 0.354932196613684
30 0.356627335315138
32 0.349653892674259
34 0.348889110321813
36 0.323337069500384
38 0.31789817244335
40 0.325061126374997
42 0.291453865244635
44 0.273518749234629
46 0.290522253000609
48 0.283100911762874
50 0.271904169605786
52 0.250336644689226
54 0.234815314162288
56 0.251332523237899
58 0.235541770491906
60 0.237025015729013
62 0.23839566877415
64 0.21588473655318
66 0.220148805599988
68 0.209167470824348
70 0.213448561465776
72 0.194282285005723
74 0.215093527633032
76 0.190305177964701
78 0.192691238276142
80 0.177503688182906
82 0.180152159828181
84 0.178350921615781
86 0.175959169947223
88 0.163428908353976
90 0.165134666052363
92 0.165809979859497
94 0.156704943513989
96 0.156689459781445
98 0.153855197121743
100 0.145407501454105
102 0.147486172530255
104 0.147854744229954
106 0.134399147942599
108 0.128049306371372
110 0.135328121574964
112 0.122340126285046
114 0.123048056402496
116 0.132003520348321
118 0.10955535242927
120 0.111469465968158
122 0.112245966794442
124 0.107980376225672
126 0.111746741982281
128 0.105706230679147
130 0.100836477657352
132 0.103703778706382
134 0.0949268437503111
136 0.0953330774360373
138 0.0976418603586002
140 0.0912737467615029
142 0.0858218941715299
144 0.0814604070884501
146 0.078152573009878
148 0.0882477807491404
150 0.0702347256847954
152 0.0669395890574299
154 0.0787715738327235
156 0.0653819425994653
158 0.0610075049487659
160 0.0587888222126365
162 0.0529365398147702
164 0.0505469381943489
166 0.0446680394441303
168 0.0424562004121941
170 0.0394715575132871
172 0.0335290494943086
174 0.0288899313599297
176 0.0250480073809904
178 0.0164282186796378
180 9.21855332002574e-16
};
\end{axis}

\end{tikzpicture}}}\vspace{-.3cm}
  \caption{Monte-Carlo simulation of $||\matr{I} - \matr{R}_{K}||_{\rm 2}$ as $n \to N$, where the corresponding $\matr{U}_{o, K} = \matr{Q}_{K}\matr{R}_{K}$ is from $\matr{L}_{\rm norm}$ of a graph generated by $\text{\text{SBM}}(180, K, \log N/N, 4\log N/N)$, with $K \in \{2, 3, 4\}.$}\vspace{-.4cm} \label{fig:Rk}
\end{figure}
Lastly, we remark that the proof for Theorem~\ref{th:main} is adapted from our prior works \cite{zhang2023detecting, wai2022partial} which applied \cite{rohe2011sbm, bunea2015covariance}. It can be found in the online appendix\footnote{\url{https://www1.se.cuhk.edu.hk/~htwai/pdf/sam24-appendix.pdf}}.

\section{Numerical Experiments} \label{sec:exp}
This section presents numerical experiments to validate our findings. We first evaluate the direct detection performance in tackling Problem~\ref{prob:detect}, then we consider an application on robustifying the blind community detection method. 

\subsection{Detecting Low-pass Signals from Partial Observations}\label{sec:ex1}
We use synthetic data to evaluate the performance of our proposed detector in various settings. In the following experiment, the graph $G$ with $N=150$ nodes and $K=3$ blocks is generated according to H\ref{assu:sbm} such that $G \sim \text{\text{SBM}}(150, 3, \log N/N, 4 \log N/N)$. The full graph signals in \eqref{eq:network_process} are generated with $\matr{x} \in \mathbb{R}^{N} \sim \mathcal{N}(\matr{0}, \matr{I})$ and $\matr{w} \sim \mathcal{N}(\matr{0}, \sigma^2\matr{I})$ where $\sigma^2 = 10^{-2}$, then we select $n$ nodes uniformly at random to form the partial observations \eqref{eq:yo}. We benchmark Algorithm~\ref{alg:detector} in distinguishing signals generated by a low-pass filter $e^{-\tau \matr{L}_{\rm norm}}$ (null hypothesis $\mathcal{T}_{0}$) from signals by a non-low-pass filter $e^{\tau \matr{L}_{\rm norm}}$ (alternative hypothesis $\mathcal{T}_{1}$), where the sharpness of the filter $\eta$ decreases as $\tau > 0$ increases. The performance is measured by the area under ROC (AUROC) such that ${\sf AUROC}=1$ when the detection is perfect. 
Figure \ref{fig:DetectLP} reports the results from 1000 Monte-Carlo trials.

We observe that the performance improves as $n,M$ increases, as well as the sharpness parameter $\eta$ controlled by $\tau$. Moreover, Algorithm~\ref{alg:detector} delivers reliable performance (with ${\sf AUROC} \approx 1$) when $n \geq 100, M \geq 100$. This indicates that the spectral pattern of low-pass graph signals are significant enough despite that $1/3$ of the nodes are not observed and only $M \approx n$ samples are observed. The above observations coincide with our finite-sample analysis in Theorem~\ref{th:main}.

\pgfplotsset{every tick label/.append style={font=\small}}
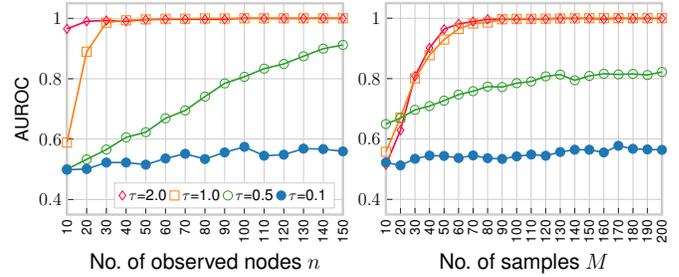
\begin{figure}[t]
\centering
  {\sf \resizebox{1.\linewidth}{!}{
\begin{tikzpicture}

\definecolor{darkgray176}{RGB}{176,176,176}
\definecolor{goldenrod1911910}{RGB}{191,191,0}
\definecolor{green01270}{RGB}{0,127,0}
\definecolor{color0}{rgb}{0.12156862745098,0.466666666666667,0.705882352941177}
\definecolor{color1}{rgb}{1,0.498039215686275,0.0549019607843137}
\definecolor{color2}{rgb}{0.172549019607843,0.627450980392157,0.172549019607843}
\definecolor{color3}{rgb}{1.0, 0.01, 0.24}



\begin{groupplot}[group style={group size=2 by 1}]
\nextgroupplot[
axis line style={white!80!black},
legend cell align={left},
legend style={
legend columns=-1,
  fill opacity=0.7,
  draw opacity=0.9,
  text opacity=1.0,
  legend pos = south east,
  draw=white!80!black,
  font = \normalsize
},
tick align=inside,
x grid style={white!80!black},
xlabel={\Large No. of observed nodes \(\displaystyle n\)},
xmajorgrids,
xmajorticks=true,
xmin=10, xmax=150,
xminorgrids,
xtick style={color=white!15!black},
xtick={10,20,30,40,50,60,70,80,90,100,110,120,130,140,150},
xticklabels={10,20,30,40,50,60,70,80,90,100,110,120,130,140,150},
xticklabel style={rotate=90},
xtick pos=left,
y grid style={white!80!black},
ylabel={\large AUROC},
ymajorgrids,
ymajorticks=true,
ymin=0.35, ymax=1.05,
yminorgrids,
ytick style={color=white!15!black},
yticklabel style={font=\large},
ytick pos=left,
width = 8cm,
height = 6.5cm,
axis line style = ultra thick
]
\addlegendentry{$\tau$=2.0}
\addplot [draw=color3, fill=color3, mark=diamond, only marks, mark size = 3]
table{%
x  y
10 0.965
20 0.9905
30 0.993
40 0.9905
50 0.9955
60 0.997
70 0.997
80 0.997
90 0.997
100 1
110 0.9995
120 0.9995
130 0.9995
140 1
150 0.9995
};

\addplot [line width=1.08pt, color3, forget plot]
table {%
10 0.965
20 0.9905
30 0.993
40 0.9905
50 0.9955
60 0.997
70 0.997
80 0.997
90 0.997
100 1
110 0.9995
120 0.9995
130 0.9995
140 1
150 0.9995
};

\addlegendentry{$\tau$=1.0}
\addplot [draw=color1, fill=color1, mark=square, only marks, mark size = 3]
table {%
10 0.5885
20 0.8895
30 0.9855
40 0.9935
50 0.9965
60 0.9985
70 0.9985
80 0.9995
90 0.9995
100 1
110 1
120 1
130 1
140 1
150 1
};

\addplot [line width=1.08pt, color1, forget plot]
table {%
10 0.5885
20 0.8895
30 0.9855
40 0.9935
50 0.9965
60 0.9985
70 0.9985
80 0.9995
90 0.9995
100 1
110 1
120 1
130 1
140 1
150 1
};

\addlegendentry{$\tau$=0.5}
\addplot [draw=color2, fill=color2, mark=o, only marks, mark size = 3]
table {%
10 0.5005
20 0.533
30 0.5655
40 0.606
50 0.623
60 0.669
70 0.695
80 0.741
90 0.7845
100 0.8065
110 0.8335
120 0.849
130 0.8745
140 0.8995
150 0.912
};

\addplot [line width=1.08pt, color2, forget plot]
table {%
10 0.5005
20 0.533
30 0.5655
40 0.606
50 0.623
60 0.669
70 0.695
80 0.741
90 0.7845
100 0.8065
110 0.8335
120 0.849
130 0.8745
140 0.8995
150 0.912
};

\addlegendentry{$\tau$=0.1}
\addplot [draw=color0, fill=color0, mark=*, only marks, mark size = 3]
table {%
10 0.4985
20 0.501
30 0.5225
40 0.5225
50 0.5155
60 0.536
70 0.5515
80 0.534
90 0.556
100 0.5745
110 0.545
120 0.5485
130 0.5685
140 0.567
150 0.5595
};

\addplot [line width=1.08pt, color0, forget plot]
table {%
10 0.4985
20 0.501
30 0.5225
40 0.5225
50 0.5155
60 0.536
70 0.5515
80 0.534
90 0.556
100 0.5745
110 0.545
120 0.5485
130 0.5685
140 0.567
150 0.5595
};

\nextgroupplot[
axis line style={white!80!black},
legend cell align={left},
legend style={
  fill opacity=0.7,
  draw opacity=0.9,
  text opacity=1.0,
  at={(0.75,0.05)},
  anchor=south west,
  draw=white!80!black,
  font = \normalsize
},
tick align=inside,
x grid style={white!80!black},
xlabel={\Large No. of samples \(\displaystyle M\)},
xmajorgrids,
xmajorticks=true,
xmin=10, xmax=200,
xminorgrids,
xtick style={color=white!15!black},
xtick={10,20,30,40,50,60,70,80,90,100,110,120,130,140,150,160,170,180,190,200},
xticklabels={10,20,30,40,50,60,70,80,90,100,110,120,130,140,150,160,170,180,190,200},
xticklabel style={rotate=90},
xtick pos=left,
y grid style={white!80!black},
ymajorgrids,
ymajorticks=true,
ymin=0.35, ymax=1.05,
yminorgrids,
ytick style={color=white!15!black},
yticklabel style={font=\large},
ytick pos=left,
width = 8cm,
height = 6.5cm,
axis line style = ultra thick
]
\addplot [draw=color3, fill=color3, mark=diamond, only marks, mark size = 3]
table {%
10 0.5145
20 0.629
30 0.8085
40 0.9015
50 0.9635
60 0.9815
70 0.989
80 0.9965
90 0.9965
100 0.9965
110 0.9975
120 0.9975
130 0.9995
140 0.999
150 0.999
160 0.997
170 0.9995
180 0.9985
190 0.999
200 0.999
};

\addplot [line width=1.08pt, color3, forget plot]
table {%
10 0.5145
20 0.629
30 0.8085
40 0.9015
50 0.9635
60 0.9815
70 0.989
80 0.9965
90 0.9965
100 0.9965
110 0.9975
120 0.9975
130 0.9995
140 0.999
150 0.999
160 0.997
170 0.9995
180 0.9985
190 0.999
200 0.999
};

\addplot [draw=color1, fill=color1, mark=square, only marks, mark size = 3]
table {%
10 0.5575
20 0.6715
30 0.801
40 0.8775
50 0.9295
60 0.9645
70 0.9825
80 0.9855
90 0.9965
100 0.998
110 0.998
120 0.9985
130 0.999
140 0.9995
150 1
160 1
170 1
180 1
190 1
200 1
};

\addplot [line width=1.08pt, color1, forget plot]
table {%
10 0.5575
20 0.6715
30 0.801
40 0.8775
50 0.9295
60 0.9645
70 0.9825
80 0.9855
90 0.9965
100 0.998
110 0.998
120 0.9985
130 0.999
140 0.9995
150 1
160 1
170 1
180 1
190 1
200 1
};

\addplot [draw=color2, fill=color2, mark=o, only marks, mark size = 3]
table {%
10 0.6495
20 0.6695
30 0.6965
40 0.709
50 0.727
60 0.747
70 0.7585
80 0.7735
90 0.772
100 0.785
110 0.7905
120 0.8075
130 0.8135
140 0.7945
150 0.809
160 0.816
170 0.814
180 0.8155
190 0.812
200 0.822
};

\addplot [line width=1.08pt, color2, forget plot]
table {%
10 0.6495
20 0.6695
30 0.6965
40 0.709
50 0.727
60 0.747
70 0.7585
80 0.7735
90 0.772
100 0.785
110 0.7905
120 0.8075
130 0.8135
140 0.7945
150 0.809
160 0.816
170 0.814
180 0.8155
190 0.812
200 0.822
};

\addplot [draw=color0, fill=color0, mark=*, only marks, mark size = 3]
table {%
10 0.5225
20 0.5125
30 0.5345
40 0.545
50 0.544
60 0.537
70 0.5455
80 0.5365
90 0.534
100 0.5425
110 0.5485
120 0.5445
130 0.557
140 0.5645
150 0.5645
160 0.5555
170 0.5775
180 0.5675
190 0.5655
200 0.564
};

\addplot [line width=1.08pt, color0, forget plot]
table {%
10 0.5225
20 0.5125
30 0.5345
40 0.545
50 0.544
60 0.537
70 0.5455
80 0.5365
90 0.534
100 0.5425
110 0.5485
120 0.5445
130 0.557
140 0.5645
150 0.5645
160 0.5555
170 0.5775
180 0.5675
190 0.5655
200 0.564
};

\end{groupplot}

\end{tikzpicture}}}\vspace{-.6cm}
  \caption{Comparing low-pass detection performance against (left) no.~of observed nodes $n$ ($M=100$), (right) no.~of observed samples $M$ ($n=100$). 
  The $\tau$ setting adjusts the sharpness of graph filters $e^{-\tau \matr{L}_{\rm norm}}$ or $e^{\tau \matr{L}_{\rm norm}}$.
  }\vspace{-.2cm}
  \label{fig:DetectLP}
\end{figure}
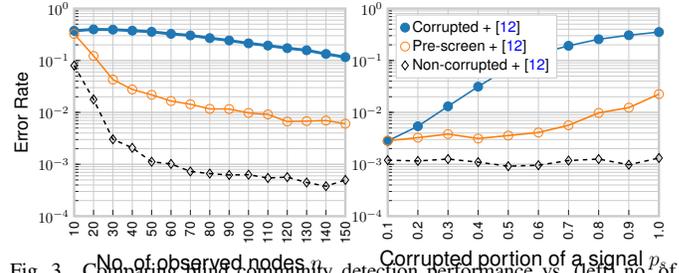
\begin{figure}[t]
\centering
  {\sf \resizebox{1.\linewidth}{!}{
\begin{tikzpicture}

\definecolor{darkgray176}{RGB}{176,176,176}
\definecolor{goldenrod1911910}{RGB}{191,191,0}
\definecolor{green01270}{RGB}{0,127,0}
\definecolor{color0}{rgb}{0.12156862745098,0.466666666666667,0.705882352941177}
\definecolor{color1}{rgb}{1,0.498039215686275,0.0549019607843137}
\definecolor{color2}{rgb}{0.172549019607843,0.627450980392157,0.172549019607843}
\definecolor{color3}{rgb}{1.0, 0.01, 0.24}



\begin{groupplot}[group style={group size=2 by 1}]
\nextgroupplot[
axis line style={white!80!black},
legend cell align={left},
legend style={
legend columns = -1,
  fill opacity=0.7,
  draw opacity=0.9,
  text opacity=1.0,
  legend pos = north east,
  at={(1.9,1.2)},
  draw=white!80!black,
  font = \normalsize
},
tick align=inside,
x grid style={white!80!black},
xlabel={\Large No. of observed nodes \(\displaystyle n\)},
xmajorgrids,
xmajorticks=true,
xmin=10, xmax=150,
xminorgrids,
xtick style={color=white!15!black},
xtick={10,20,30,40,50,60,70,80,90,100,110,120,130,140,150},
xticklabels={10,20,30,40,50,60,70,80,90,100,110,120,130,140,150},
xticklabel style={rotate=90},
xtick pos=left,
y grid style={white!80!black},
ylabel={\large Error Rate},
ymajorgrids,
ymajorticks=true,
ymin=0.0001, ymax=1.0,
yminorgrids,
ymode=log,
ytick style={color=white!15!black},
yticklabel style={font=\normalsize},
ytick pos=left,
width = 8cm,
height = 6.5cm,
axis line style = ultra thick
]

\addplot [draw=color0, fill=color0, mark=*, only marks, mark size = 3] table{%
10	0.373
20	0.3968
30	0.392833333333333
40	0.37555
50	0.3567
60	0.325366666666667
70	0.304157142857143
80	0.2697125
90	0.242744444444444
100	0.21411
110	0.193672727272727
120	0.172625
130	0.157353846153846
140	0.133521428571429
150	0.11592
};

\addplot [line width=2pt, color0, forget plot] table{%
10	0.373
20	0.3968
30	0.392833333333333
40	0.37555
50	0.3567
60	0.325366666666667
70	0.304157142857143
80	0.2697125
90	0.242744444444444
100	0.21411
110	0.193672727272727
120	0.172625
130	0.157353846153846
140	0.133521428571429
150	0.11592
};

\addplot [draw=color1, fill=color1, mark=o, only marks, mark size = 3] table{%
10	0.3247
20	0.12255
30	0.0428666666666666
40	0.027625
50	0.02178
60	0.01655
70	0.0143428571428571
80	0.011675
90	0.0116
100	0.00978999999999999
110	0.00915454545454545
120	0.00665833333333333
130	0.00677692307692307
140	0.00695714285714285
150	0.00606
};

\addplot [line width=1pt, color1, forget plot] table{%
10	0.3247
20	0.12255
30	0.0428666666666666
40	0.027625
50	0.02178
60	0.01655
70	0.0143428571428571
80	0.011675
90	0.0116
100	0.00978999999999999
110	0.00915454545454545
120	0.00665833333333333
130	0.00677692307692307
140	0.00695714285714285
150	0.00606
};

\addplot [draw=black, fill=color2, mark=diamond, only marks, mark size = 3] table{%
10	0.0798
20	0.0178
30	0.00303333333333333
40	0.002075
50	0.00112
60	0.00101666666666667
70	0.000728571428571429
80	0.0006625
90	0.000622222222222222
100	0.00063
110	0.000545454545454545
120	0.000566666666666667
130	0.000446153846153846
140	0.000378571428571429
150	0.0005
};

\addplot [line width=1pt, black, dashed, forget plot] table{%
10	0.0798
20	0.0178
30	0.00303333333333333
40	0.002075
50	0.00112
60	0.00101666666666667
70	0.000728571428571429
80	0.0006625
90	0.000622222222222222
100	0.00063
110	0.000545454545454545
120	0.000566666666666667
130	0.000446153846153846
140	0.000378571428571429
150	0.0005
};

\nextgroupplot[
axis line style={white!80!black},
legend cell align={left},
legend style={
  fill opacity=0.7,
  draw opacity=0.9,
  text opacity=1.0,
  legend pos = north west,
  draw=white!80!black,
  font = \normalsize
},
tick align=inside,
x grid style={white!80!black},
xlabel={\Large Corrupted portion of a signal $p_s$},
xmajorgrids,
xmajorticks=true,
xmin=0.1, xmax=1,
xminorgrids,
xtick style={color=white!15!black},
xtick={0.1,0.2,0.3,0.4,0.5,0.6,0.7,0.8,0.9,1},
xticklabels={0.1,0.2,0.3,0.4,0.5,0.6,0.7,0.8,0.9,1.0},
xticklabel style={rotate=90},
xtick pos=left,
y grid style={white!80!black},
ymajorgrids,
ymajorticks=true,
ymin=0.0001, ymax=1.0,
yminorgrids,
ymode=log,
ytick style={color=white!15!black},
yticklabel style={font=\normalsize},
ytick pos=left,
width = 8cm,
height = 6.5cm,
axis line style = ultra thick
]

\addlegendentry{Corrupted +\cite{wai2022partial}}
\addplot [draw=color0, fill=color0, mark=*, only marks, mark size = 3]
table {%
0.1	0.00282
0.2	0.0054
0.3	0.01308
0.4	0.0312999999999999
0.5	0.0709599999999998
0.6	0.1333
0.7	0.1918
0.8	0.25634
0.9	0.30634
1	0.35214
};
\addplot [line width=1.08pt, color0, forget plot]
table {%
0.1	0.00282
0.2	0.0054
0.3	0.01308
0.4	0.0312999999999999
0.5	0.0709599999999998
0.6	0.1333
0.7	0.1918
0.8	0.25634
0.9	0.30634
1	0.35214
};

\addlegendentry{Pre-screen +\cite{wai2022partial}}
\addplot [draw=color1, fill=color1, mark=o, only marks, mark size = 3]
table {%
0.1	0.00284
0.2	0.00326
0.3	0.00382
0.4	0.00314
0.5	0.00358
0.6	0.0041
0.7	0.00564
0.8	0.00978
0.9	0.01234
1	0.02248
};
\addplot [line width=1.08pt, color1, forget plot]
table {%
0.1	0.00284
0.2	0.00326
0.3	0.00382
0.4	0.00314
0.5	0.00358
0.6	0.0041
0.7	0.00564
0.8	0.00978
0.9	0.01234
1	0.02248
};

\addlegendentry{Non-corrupted +\cite{wai2022partial}}
\addplot [draw=black, fill=black, mark=diamond, only marks, mark size = 3]
table {%
0.1	0.0012
0.2	0.00116
0.3	0.00126
0.4	0.0011
0.5	0.00092
0.6	0.00096
0.7	0.00118
0.8	0.00126
0.9	0.00098
1	0.00132
};

\addplot [line width=1pt, black, dashed, forget plot]
table {%
0.1	0.0012
0.2	0.00116
0.3	0.00126
0.4	0.0011
0.5	0.00092
0.6	0.00096
0.7	0.00118
0.8	0.00126
0.9	0.00098
1	0.00132
};

\end{groupplot}

\end{tikzpicture}}}\vspace{-.6cm}
  \caption{Comparing blind community detection performance vs. (left) no.~of observed nodes $n$ ($p_s = 1$), (right) corrupted portion of signals $p_s$ ($n=50$).}\vspace{-.2cm}
  \label{fig:BlindCD}
\end{figure}

\subsection{Application: Robustifying Blind Community Detection}
We illustrate an application of Algorithm~\ref{alg:detector} as a pre-screening procedure before applying prior work that demands low-pass graph signals. We consider the blind community detection method \cite{wai2022partial} which directly infer communities in a graph from low-pass graph signals that are partially observed. To satisfy the low-pass graph signal requirement, we apply Algorithm~\ref{alg:detector} on \emph{small batches} of $M_{\sf batch}$ graph signal observations and retain (resp.~drop) the small batches that are identified as low-pass (resp.~non-low-pass). 
The pre-screened dataset is then provided to \cite{wai2022partial} to infer the communities.

We consider $G \sim \text{SBM}(150, 3, \log N/N, 7\log N/N$), with $N = 150$ nodes and $K = 3$ clusters. The \emph{normal} graph signals are generated using \eqref{eq:network_process}, \eqref{eq:yo} with $\sigma^2 = 10^{-2}$ and the filter ${\cal H}( \matr{S} ) = (\matr{I}-0.5\matr{L}_{\rm norm})^{3}$, where $10\%$ of samples are \emph{corrupted} in a burst of length $m_{\sf burst} = 10$, such that $p_s$-fraction of nodal observations are replaced with {${\cal N}(0,1)$}. For the pre-screening procedure, we apply Algorithm~\ref{alg:detector} on small batches of size $M_{\sf batch} = 50$ from {$M=10^3$ samples}, with $\delta=0.5$. Figure \ref{fig:BlindCD} reports the results of 1000 Monte-Carlo trials.

Observe the dataset corruption severely affects the performance of blind community detection \cite{wai2022partial}. Meanwhile, our pre-screening procedure robustifies the method in \cite{wai2022partial}. We note the effectiveness of pre-screening improves with $n$ as it approaches the performance of non-corrupted dataset, coinciding with Theorem~\ref{th:main} that low-pass detection becomes more accurate as $n$ increases. Pre-screening also delivers consistent improvement across different levels of signal corruption.\vspace{.2cm}

\noindent \textbf{Conclusions.} This paper studies the low-pass graph signal detection problem with partial observations. We showed that a simple $K$-means score detector can distinguish spectral pattern of the low-pass/non-low-pass signals and analyzed its sample complexity. Our work can robustify GSP on partially observed signals. Future work includes deriving an explicit bound w.r.t.~no.~of observed nodes $n$ and explore other applications.

\newpage 
\bibliographystyle{IEEEtran}
\bibliography{ref}

\newpage




\section*{Appendix: Proof of Theorem~\ref{th:main}}
{\bf Additional Notations.} In the following analysis, we define the QR decomposition of $\matr{U}_{o, K}$ as $\matr{U}_{o, K} = c_0 \matr{Q}_{K}\matr{R}_{K}$, in which $c_0 = \sqrt{n/N}$ is a normalization parameter, $\matr{Q}_{K } \in \mathbb{R}^{n \times K}$ is an orthogonal matrix spanning the range of $\matr{U}_{o, K}$ and $\matr{R}_{K}$ is upper-triangular. We also set the diagonal of $\matr{R}_K$ to be non-negative. Moreover, under H\ref{assu:sbm}, $\mathbf{A}$ satisfies $\mathbb{E}[\mathbf{A}] = \mathbf{Z}\mathbf{B}\mathbf{Z}^\top =: \mathcal{A}$; the population normalized Laplacian matrix of the \text{SBM} is then $\mathcal{L}_{\rm norm} = \matr{I} - \mathcal{D}^{-1/2}\mathcal{A}\mathcal{D}^{-1/2}$, where $\mathcal{D} = \diag(\sum_{j=1}^{N}\mathcal{A}_{1j}, ..., \sum_{j=1}^{N}\mathcal{A}_{Nj})$.
Lastly, the set $\mathcal{R}_{K}^{m \times n}$ consists of $m \times n$ matrices having at most $K$ unique rows. 

Our proof is adapted from \cite{zhang2023detecting, wai2022partial}. First, consider the ground truth as $\mathcal{T}_{\rm gnd} = \mathcal{T}_{0}$, which implies $\matr{U}_{K} = \matr{V}_{K}\matr{\Pi}$ for some permutation matrix $\matr{\Pi}$. Define the indicator matrix $\matr{X}^{*}$ be associated with the partition $\mathcal{C}^{*} \in  \arg\min_{\mathcal{C}}\mathbb{K}(\matr{Q}_{K}, \mathcal{C})$:
\begin{align*}
    \matr{X}^{*}_{ik} := \begin{cases} 1/\sqrt{|\mathcal{C}^{*}_{i}|} & \text{if } i \in \mathcal{C}^{*}_{i} \text{,}\\ 0 & \text{otherwise.}\end{cases}
\end{align*}
We observe that
\begin{align*}
    &\sqrt{\mathbb{K}^{*}(\widehat{\matr{Q}}_{K})}
    \leq ||(\matr{I} - \matr{X}^{*}(\matr{X}^{*})^{\top})\widehat{\matr{Q}}_{K}||_{\rm F}\\
    &= ||(\matr{I} - \matr{X}^{*}(\matr{X}^{*})^{\top})\widehat{\matr{Q}}_{K}\widehat{\matr{Q}}_{K}^{\top}||_{\rm F}\\
    &\leq ||(\matr{I} - \matr{X}^{*}(\matr{X}^{*})^{\top})\matr{Q}_{K}\matr{Q}_{K}^{\top}||_{\rm F} + ||\matr{Q}_{K}\matr{Q}_{K}^{\top} - \widehat{\matr{Q}}_{K}\widehat{\matr{Q}}_{K}||_{\rm F}\\
    &= \sqrt{\mathbb{K}^{*}(\matr{Q}_{K})} + ||\matr{Q}_{K}\matr{Q}_{K}^{\top} - \widehat{\matr{Q}}_{K}\widehat{\matr{Q}}_{K}||_{\rm F}.
\end{align*}
Similarly, we further have
\begin{align*}
    &\sqrt{\mathbb{K}^{*}(\matr{Q}_{K})}
    \leq {c_0^{-1}} \sqrt{\mathbb{K}^{*}(\matr{U}_{o, K})} + ||\matr{Q}_{K} - {c_0^{-1}} \matr{U}_{o, K}||_{\rm F}\\
    &\leq {c_0^{-1}} \sqrt{\mathbb{K}^{*}(\matr{U}_{o, K})} + ||\matr{Q}_{K}||_{\rm F}||\matr{I} - \matr{R}_{K}||_{\rm 2} \\
    &= {c_{0}^{-1}} \sqrt{\mathbb{K}^{*}(\matr{U}_{o, K})} + \sqrt{K}||\matr{I} - \matr{R}_{K}||_{\rm 2}.
\end{align*}
Define the orthogonal matrix $\overline{\mathcal{O}}_{K} = \mathcal{O}_{K}\matr{\Pi}$, where $\mathcal{O}_{K}$ is from Lemma \ref{lemma:full_eigenvectors}. Since $\mathcal{V}_{K}\overline{\mathcal{O}}_{K} \in \mathcal{R}_{K}^{N \times K}$ \cite{rohe2011sbm}, we have $\matr{E}_{o}\mathcal{V}_{K}\overline{\mathcal{O}}_{K} \in \mathcal{R}_{K}^{n \times K}$. Consequently, by H\ref{assu:sbm} and Lemma \ref{lemma:full_eigenvectors}, with probability at least $1-2/N$, 
\begin{align*}
    &\sqrt{\mathbb{K}^{*}(\matr{U}_{o, K})} = \min_{\overline{\matr{U}} \in \mathcal{R}_{K}^{n \times K}}||\matr{U}_{o, K} - \overline{\matr{U}}||_{\rm F}\\
    &\leq ||\matr{U}_{o, K} - \matr{E}_{o}\mathcal{V}_{K}\overline{\mathcal{O}}_{K}||_{\rm F} \leq ||\matr{E}_{o}||_{2}||\matr{U}_{K} - \mathcal{V}_{K}\overline{\mathcal{O}}_{K}||_{\rm F}\\
    &= ||\matr{U}_{K} - \mathcal{V}_{K}\overline{\mathcal{O}}_{K}||_{\rm F} \leq \frac{35\sqrt{K^3 \log N}}{\sqrt{p(N-K)}}.
\end{align*}
Combining the upper-bound of $\sqrt{\mathbb{K}^{*}(\matr{Q}_{K})}$ with Lemma \ref{lemma:qk} as well as H\ref{assu:spectral_gap}, H\ref{assu:sharp}, we conclude that when the null hypothesis holds, with probability at least $1 - 2/N - 5/M - \delta_{\rm gap}$,
\begin{align}
    &\mathbb{K}^{*}(\widehat{\matr{Q}}_{K}) \leq \Bigg[ \sqrt{\frac{N}{n}} \frac{35\sqrt{K^3 \log N}}{\sqrt{p(N-K)}} + \sqrt{K}||\matr{I} - \matr{R}_{K}||_{\rm 2} \nonumber\\
    &+2\sqrt{K}\Bigg(3\gamma\eta
    + \frac{c_1 \tr(\overline{\matr{C}}_{o})\sqrt{2\log M/M} + \sigma^2}{\rho_{\rm gap}}\Bigg)\Bigg]^2. \label{eq:upper_bound}
\end{align}

The next case is to consider the ground truth as $\mathcal{T}_{gnd} = \mathcal{T}_1$. 
Define $\widehat{\matr{X}}$ associated with $\widehat{\mathcal{C}} \in \arg\min_{\mathcal{C}}\mathbb{K}(\widehat{\matr{Q}}_{K}, \mathcal{C})$. Similar to the previous case, we have
\begin{align*}
    &\sqrt{\mathbb{K}^{*}(\matr{Q}_{K})}
    \leq ||(\matr{I} - \widehat{\matr{X}}\widehat{\matr{X}}^{\top})\matr{Q}_{K}||_{\rm F}\\
    &= ||(\matr{I} - \widehat{\matr{X}}\widehat{\matr{X}}^{\top})\matr{Q}_{K}\matr{Q}_{K}^{\top}||_{\rm F}\\
    &\leq ||(\matr{I} - \widehat{\matr{X}}\widehat{\matr{X}}^{\top})\widehat{\matr{Q}}_{K}\widehat{\matr{Q}}_{K}^{\top}||_{\rm F} + ||\matr{Q}_{K}\matr{Q}_{K}^{\top} - \widehat{\matr{Q}}_{K}\widehat{\matr{Q}}_{K}||_{\rm F}\\
    &= \sqrt{\mathbb{K}^{*}(\widehat{\matr{Q}}_{K})} + ||\matr{Q}_{K}\matr{Q}_{K}^{\top} - \widehat{\matr{Q}}_{K}\widehat{\matr{Q}}_{K}||_{\rm F},
\end{align*}
which implies $\sqrt{\mathbb{K}^{*}(\widehat{\matr{Q}}_{K})} \geq \sqrt{\mathbb{K}^{*}(\matr{Q}_{K})} - ||\matr{Q}_{K}\matr{Q}_{K}^{\top} - \widehat{\matr{Q}}_{K}\widehat{\matr{Q}}_{K}||_{\rm F}$. By the same technique, we  obtain $\sqrt{\mathbb{K}^{*}(\matr{Q}_K)} \geq c_{0}^{-1}\sqrt{\mathbb{K}^{*}(\matr{U}_{o, K})} - \sqrt{K}||\matr{I} - \matr{R}_{K}||_{\rm 2}$.

Our remaining task is to lower bound $\sqrt{\mathbb{K}^{*}(\matr{U}_{o, K})}$ using H\ref{assu:csbm}. Let $\matr{U}_{r, s} = [\matr{u}_r, ..., \matr{u}_{s}]$ consist of column vectors from $\matr{U}$, with $r \leq s$. Also, let $\pi$ be a permutation function on $\{1, ..., N\}$, satisfying $|h_i| = |h(\lambda_{\pi(i)})|$. We can see that the set $\mathcal{P} := \{i: 1 \leq i \leq K, K+1 \leq \pi(i) \leq N\}$ is non-empty under $\mathcal{T}_{gnd} = \mathcal{T}_{1}$. Then, for any $r \leq s$ such that $[r, s] \in \mathcal{P}$, by Lemma \ref{lemma:obs},
\begin{align*}
    \mathbb{K}^{*}(\matr{U}_{o, K})
    &\geq \mathbb{K}^{*}(\matr{U}_{K}) - |\mathbb{K}^{*}(\matr{U}_{o, K}) - \mathbb{K}^{*}(\matr{U}_{K})|\\
    &\geq \mathbb{K}^{*}(\matr{U}_{K}) - \frac{2450 K^3 \log N}{p(N-K)}\\
    &\geq \mathbb{K}^{*}(\matr{U}_{r, s}) - \frac{2450 K^3 \log N}{p(N-K)}.
\end{align*}
By H\ref{assu:csbm}, we have $\mathbb{K}^{*}(\matr{U}_{r, s}) \geq c_{\text{\text{SBM}}}$. Together with H\ref{assu:spectral_gap}, H\ref{assu:sharp}, with probability at least $1 - 4/N - 5/M - \delta_{\rm gap} - \delta_{\rm \text{SBM}}$, the following lower bound holds
\begin{align}
    &\mathbb{K}^{*}(\widehat{\matr{Q}}_{K}) \geq \Bigg[\sqrt{\frac{N}{n}}\sqrt{c_{\rm \text{SBM}} - \frac{2450 K^3 \log N}{p(N-K)}}- \sqrt{K}||\matr{I} - \matr{R}_{K}||_{\rm 2} \nonumber\\
    &-2\sqrt{K}\Bigg(3\gamma\eta
    + \frac{c_1 \tr(\overline{\matr{C}}_{o})\sqrt{2\log M/M} + \sigma^2}{\rho_{\rm gap}}\Bigg)\Bigg]^2. \label{eq:lower_bound}
\end{align}

Finally, we can conclude the proof by noting that $\widehat{\mathcal{T}} = \mathcal{T}_{\rm gnd}$ holds when $\delta$ upper bounds the right-hand side of \eqref{eq:upper_bound} and also lower bounds the right-hand side of \eqref{eq:lower_bound}.

\subsection*{Technical Lemmas}
\begin{Lemma} \label{lemma:obs}
Under H\ref{assu:sbm}. Let $\matr{U}_{K}$ denote the columns of the first $K$ eigenvectors of $\matr{L}_{\rm norm}$, and $\matr{U}_{o, K} = \matr{E}_{o}\matr{U}_{K}$.
With probability at least $1 - 4/N$,
\begin{align*}
    |\mathbb{K}^{*}(\matr{U}_{o, K}) - \mathbb{K}^{*}(\matr{U}_{K})| \leq \frac{2450 K^3 \log N}{p(N-K)}.
\end{align*}
\end{Lemma}
\begin{proof}
By the triangular inequality, 
\[
|\mathbb{K}^{*}(\matr{U}_{o, K}) - \mathbb{K}^{*}(\matr{U}_{K})| \leq \mathbb{K}^{*}(\matr{U}_{o, K}) + \mathbb{K}^{*}(\matr{U}_{K}).
\]
Applying Lemma \ref{lemma:kmeans} yields $\mathbb{K}^{*}(\matr{U}_{K}) \leq \frac{35^2 K^3 \log N}{pN(N-K)}$ with probability at least $1-2/N$. 

We now derive an upper-bound for $\mathbb{K}^{*}(\matr{U}_{o, K})$. As $\mathcal{V}_{K}\mathcal{O}_{K} \in \mathcal{R}_{K}^{N \times K}$, we have $\matr{E}_{o}\mathcal{V}_{K}\mathcal{O}_{K} \in \mathcal{R}_{K}^{n \times K}$. Then, with probability at least $1-2/N$,
\begin{align*}
    &\mathbb{K}^{*}(\matr{U}_{o, K}) \leq ||\matr{U}_{o, K} - \matr{E}_{o}\mathcal{V}_{K}\mathcal{O}_{K}||_{\rm F}^{2} \leq ||\matr{E}_{o}||_{\rm 2}^2||\matr{U}_{K} - \mathcal{V}_{K}\mathcal{O}_{K}||_{\rm F}^{2}\\
    &= ||\matr{U}_{K} - \mathcal{V}_{K}\mathcal{O}_{K}||_{\rm F}^{2} = \frac{35^2 K^3 \log N}{p(N-K)}.
\end{align*}
This concludes the proof.
\end{proof}

\begin{Lemma} \label{lemma:qk}
Under H\ref{assu:spectral_gap}, H\ref{assu:sharp}, the following inequality holds with probability at least $1-5/M$
\begin{align*}
    &||\matr{Q}_{K}\matr{Q}_{K}^{\top} - \widehat{\matr{Q}}_{K}\widehat{\matr{Q}}_{K}||_{\rm F}\leq\\
    &2\sqrt{K}\Bigg(3\gamma\eta
    + \frac{c_1 \tr(\overline{\matr{C}}_{o})\sqrt{2\log M/M} + \sigma^2}{\rho_{\rm gap}}\Bigg).
\end{align*}
where $c_1$ is a constant independent of $N, M$ \cite{bunea2015covariance}.  
\end{Lemma}
\begin{proof}
By \cite[Proposition 1]{wai2022partial}, we have a deterministic upper-bound:
\begin{align*}
    &||\matr{Q}_{K}\matr{Q}_{K}^{\top} - \widehat{\matr{Q}}_{K}\widehat{\matr{Q}}_{K}||_{\rm F} \leq\\
    &\sqrt{2K}\left(\sqrt{2}\gamma(2||\matr{U}_{K}||_{2} + ||\matr{U}_{N-K}||_{2})\eta_{K} + \frac{||\widehat{\matr{C}}_{o} - \overline{\matr{C}}_{o}||_{\rm 2}}{\rho_{\rm gap}}\right)
\end{align*}
where we further have $||\matr{U}_{K}||_{2} \leq 1$ and $||\matr{U}_{N - K}||_{2} \leq 1$ due to their orthogonality. In addition, applying \cite[Theorem 2.1]{bunea2015covariance} on $\{\matr{y}_{o, m}\}_{m=1}^{M}$ yields the following inequality: with probability at least $1 - 5/M$,
\begin{align*}
    ||\widehat{\matr{C}}_{o} - \overline{\matr{C}}_{o}||_{\rm 2} \leq 2 c_1 \tr(\overline{\matr{C}}_{o})\sqrt{\frac{\log M}{M}} + \sigma^2.
\end{align*}
This concludes the proof.
\end{proof}

The last two auxiliary lemmas are borrowed from \cite{zhang2023detecting}, which have been inspired by \cite{rohe2011sbm, bunea2015covariance}:
\begin{Lemma}[\texorpdfstring{\cite[Proposition 2]{zhang2023detecting}}{Proposition 2 by Zhang et al.}] \label{lemma:kmeans}
Under H\ref{assu:sbm}. For $\matr{V}_{K}$ consisting of $K$ bottom eigenvectors of $\matr{L}_{\rm norm}$, with probability at least $1-2/N$,
\[
\mathbb{K}^{*}(\matr{V}_{K}) \leq \frac{35^2 K^3 \log N}{p(N-K)}.
\]
\end{Lemma}

\begin{Lemma}[\texorpdfstring{\cite[Lemma 2]{zhang2023detecting}}{Lemma 2 from Zhang et al.}] \label{lemma:full_eigenvectors}
Under H\ref{assu:sbm}. Let $\matr{V}_{K}, \mathcal{V}_{K}$ denote the columns of the first $K$ eigenvectors of $\matr{L}_{\rm norm}, \mathcal{L}_{\rm norm}$. With probability at least $1 - 2/N$, there exists an orthogonal matrix $\mathcal{O}_{K} \in \mathbb{R}^{K \times K}$ such that
\[
||\matr{V}_{K} - \mathcal{V}_{K}\mathcal{O}_{K}||_{\rm F} \leq \frac{35\sqrt{K^3 \log N}}{\sqrt{p(N-K)}}.
\]
\end{Lemma}


\end{document}